\documentclass[11pt]{scrartcl}
\usepackage{graphicx,amssymb,amsmath,amsthm}
\usepackage{algorithm}
\usepackage{algpseudocode}

\usepackage{pgf,tikz}
\usetikzlibrary{shapes}
\usetikzlibrary{shapes.gates.logic.US}

\usepackage{color}
\usepackage{hyperref}
\usepackage{microtype}
\usetikzlibrary{shapes}
\usetikzlibrary{shapes.gates.logic.US}

\usepackage[top=1.2in, bottom=1.2in, left=1in, right=1in]{geometry}

\title{Two-Level Rectilinear Steiner Trees}

\author{Stephan Held\thanks{Research Institute for Discrete Mathematics,
        University of Bonn, \newline {\tt held@or.uni-bonn.de, kaemmerling@or.uni-bonn.de}.}
        \and
	Nicolas K\"ammerling\footnotemark[1]}
\date{~\vspace{-1cm}}        
\newcommand{\Rset}{ {\mathbb{R}} }

\index{Held, Stephan}
\index{K\"ammerling, Nicolas}

\newcommand{\NP}{\mbox{\slshape NP}}

\newtheorem{theorem}{Theorem}[section]
\newtheorem{lemma}{Lemma}[section]
\newtheorem{cor}{Corollary}[section]

\def\mnorm#1{\|#1\|_1}

\definecolor{darkGreen}{rgb}{0,0.5,0}

\begin{document}
\thispagestyle{empty}
\maketitle
                                                                                                                                                                                                                                                                                                                                                                                                                                                                                                                                                                                                                                                                                                                                                                                                        
\begin{abstract}
\noindent  \textbf{Abstract}

  Given a set $P$ of terminals in the plane and a partition of $P$
  into $k$ subsets $P_1, \ldots, P_k$, a two-level rectilinear Steiner
  tree consists of a rectilinear Steiner tree $T_i$ connecting the
  terminals in each set $P_i$ ($i=1,\dots,k$) and a top-level
  tree $T_{top}$ connecting the trees $T_1, \ldots, T_k$. The goal is
  to minimize the total length of all trees. This problem arises
  naturally in the design of low-power physical implementations
  of parity functions on a computer chip.
  
  For bounded $k$ we present a polynomial time approximation scheme
  (PTAS) that is based on Arora's PTAS for rectilinear Steiner trees
  after lifting each partition into an extra dimension.
  
  For the general case we propose an algorithm that predetermines a
  connection point for each $T_i$ and $T_{top}$ ($i=1,\ldots,k$).
  Then, we  apply any approximation algorithm for minimum
  rectilinear Steiner trees in the plane to compute each $T_i$ and
  $T_{top}$ independently.
  
  This gives us a $2.37$-factor approximation with a running time of
  $\mathcal{O}(|P|\log|P|)$ suitable for fast practical computations.
  The approximation factor reduces to $1.63$ by applying Arora's
  approximation scheme in the plane.
\end{abstract}

\section{Introduction}
We consider the {\it two-level rectilinear Steiner tree problem}
(R2STP) that arises from an application in VLSI design.  Consider the
computation of a parity function of $k$ input bits using 2-input
XOR-gates. 
Due to the symmetry, associativity, and commutativity of the XOR
function, this can be realized by an arbitrary binary tree with $k$
leaves, rooted at the output, by inserting an XOR-gate at every
internal vertex \cite{xiang-etal:2010}.
Throughout this paper we consider the parity function as a placeholder
for any fan-in function of the type $x_1\circ x_2\circ\dots\circ x_k$, where
$\circ$ is a symmetric, associative, and commutative 2-input operator,
i.\ e.\  $\circ \in \{\oplus, \vee, \wedge\}$.

On a chip such a tree has to be embedded into the plane and all
connections must be realized by rectilinear segments.  If each input
and the output are single points on the chip, a realization of minimum
length and thus power consumption is given by a minimum length rectilinear Steiner
tree. This is a tree connecting the inputs and 
the output by horizontal and vertical line segments using
additional so-called Steiner vertices to achieve a shorter length than
a minimum spanning tree.  At each Steiner vertex of degree three an
XOR-gate is placed. Higher degree vertices can be dissolved into
degree three vertices sharing their position.
Figure~\ref{fig:motivating example} shows an example of an embedded
parity function on the left.

\begin{figure} [t]
\begin{minipage}[c]{0.49\textwidth}
\centering
\begin{tikzpicture}[scale=0.33, set style={{help lines}+=[white]}
  ]
  \draw[style=help lines] (-1, -1) grid +(16, 7);

  \node (R) at (15, 3) {};

  \begin{scope}[shape=rectangle,inner sep=1pt, minimum size=0.42cm, fill=darkGreen!65]
   \tikzstyle{every node}=[fill]
   \node (NR1) at ( 0  , 6) {$p_1$};
  \end{scope}
  \begin{scope}[shape=rectangle,inner sep=1pt, minimum size=0.42cm, fill=orange!65]
   \tikzstyle{every node}=[fill]
   \node (NR4) at ( 0 ,  0) {$p_2$};
  \end{scope}

  \node[scale=0.7,draw,fill=red,xor gate US] (xor) at (0.9, 3) {};

  \draw[color=red, line width=1pt]  (R) |- (xor.output);
  \draw[color=red, line width=1pt]  (xor.input 1) -| (NR1);
  \draw[color=red, line width=1pt]  (xor.input 2) -| (NR4);

  \begin{scope}[shape=rectangle,inner sep=1pt, minimum size=0.42cm, fill=blue!33]
   \tikzstyle{every node}=[fill]
   \node (S1) at ( 15  ,  3) {$p_3$};
  \end{scope}
\end{tikzpicture}
\end{minipage}
\hfill
\begin{minipage}[c]{0.49\textwidth}
\centering
\begin{tikzpicture}[scale=0.33, set style={{help lines}+=[white]}]
  \draw[style=help lines] (-1, -1) grid +(16, 7);

  \node (R) at (15, 3) {};
  \node (R1) at (12.3, 6.35) {};
  \node (R2) at (12.3, -0.38) {};

  \begin{scope}[shape=rectangle,inner sep=1pt, minimum size=0.42cm, fill=darkGreen!65]
   \tikzstyle{every node}=[fill]
   \node (NR1) at ( 0  , 6) {$p_1$};
  \end{scope}
  \begin{scope}[shape=rectangle,inner sep=1pt, minimum size=0.42cm, fill=orange!65]
   \tikzstyle{every node}=[fill]
   \node (NR4) at ( 0 ,  0) {$p_2$};
  \end{scope}
  \begin{scope}[shape=rectangle,inner sep=1pt, minimum size=0.42cm, fill=darkGreen!65]
   \tikzstyle{every node}=[fill]
   \node (N13) at ( 15, 6) {$p_1'$};
  \end{scope}
  \begin{scope}[shape=rectangle,inner sep=1pt, minimum size=0.42cm, fill=orange!65]
   \tikzstyle{every node}=[fill]
   \node (N42) at ( 15, 0) {$p_2'$};
  \end{scope}

  \node[scale=0.7,draw,fill=red,xor gate US] (xor) at (13.3, 3) {};

  \draw[color=red, line width=1pt]  (R) |- (xor.output);
  \draw[color=red, line width=1pt]  (xor.input 1) -| (R1);
  \draw[color=red, line width=1pt]  (xor.input 2) -| (R2);

  \begin{scope}[shape=rectangle,inner sep=1pt, minimum size=0.42cm, fill=blue!33]
   \tikzstyle{every node}=[fill]
   \node (S1) at ( 15  ,  3) {$p_3$};
  \end{scope}

  \draw[line width=1pt, color=darkGreen] (NR1) |- (N13) ;
  \draw[line width=1pt, color=orange] (NR4) |- (N42) ;

\end{tikzpicture}
\end{minipage}
\caption{On the left, we have two inputs $p_1$ and $p_2$ and a single output $p_3$. 
The XOR-gate should be placed at the median of the three terminals.
If the inputs have the side outputs  $p_1'$ and $p_2'$, the XOR-gate 
should be placed at $p_3$, saving  the horizontal length.}
\label{fig:motivating example}
\end{figure}
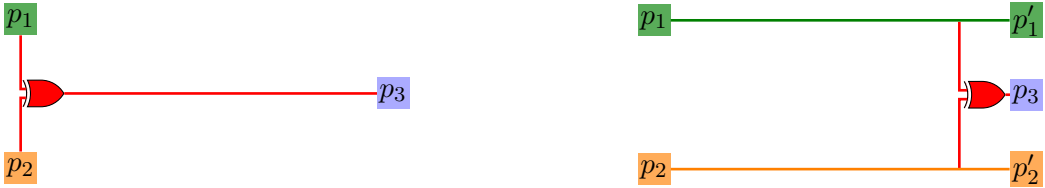

In practice input signals may be needed for other computations on the
chip and thus delivered to other side outputs. Similarly, the result
may have to be delivered to multiple output terminals.  Thus, each
input and its successors and the output terminals must be connected by
separate Steiner trees as well. These trees are then connected by a top-level
Steiner tree into which the XOR-gates will be inserted.  Considering
the additional terminals allows to construct a potentially shorter
top-level and  two-level Steiner tree  as shown in Figure~\ref{fig:motivating example} on the
right. Algorithms ignoring the side outputs cannot guarantee an approximation
factor better than two, as we will see in Section~\ref{sec:simple-bottom-up}.

This motivates the definition of the {\it minimum two-level rectilinear Steiner
  tree problem}, where we are given a set $P\subset \Rset^2$ of $n$ terminals and a partition of $P$ into $k$ subsets $P_1, \ldots, P_k$.

A {\it two-level rectilinear Steiner tree} $T = (T_{top}, T_1, \ldots,
T_k)$ consists of a Steiner tree $T_i$ for each $i \in \{1,\dots,k\}$
connecting the terminals in $P_i$ and a (group) Steiner tree $T_{top}$ 
connecting the embedded trees $\{T_1, \ldots, T_k\}$. 
We call $T_{top}$ the {\it top-level tree}. 
Note that all trees are allowed to cross.
The objective  is to
minimize the total length of all trees
  $$l(T):= \sum_{i=1}^{k}l(T_i) +l(T_{top}),$$
where $l(T'):= \sum_{\{x,y\}\in E(T')} \mnorm{x-y}$ is the
$\ell_1$-length of a Steiner tree $T'$.
%

For each $i\in \{1,\dots,k\}$  the top-level tree and $T_i$ intersect in at least one point.
We can select one such point $q_i\in T_{top}\cap T_i$ and  call it {\it connection point} for $T_i$ and $T_{top}$.
Then $T_{top}$ is a Steiner tree for the terminals $\{q_1,\dots,q_k\}$ and each $T_i$ is a Steiner
tree for $P_i\cup \{q_i\}$.

Obviously, this problem is \NP-hard as it contains the minimum
rectilinear Steiner tree problem in two ways: if $k=1$ or if $|P_i|=1$
for $i\in \{1,\dots,k\}$.  The problem is  \NP-complete, because
there is always an optimum solution in the Hanan grid \cite{hanan} of
$P$. This simple fact will arise later as a side-result in
 Corollaries~\ref{cor:contained-in-Hanan-grid} and \ref{cor:contained-in-NP}.

Designing the top-level tree as a stand-alone problem is hard.  If
all subtrees $T_i$ $(i\in \{1,\dots,k\})$ are fixed, $T_{top}$ cannot
be approximated to arbitrary quality, as the group Steiner tree
problem for connected groups in the Euclidean plane cannot be approximated within a factor of
$(2-\epsilon)$ \cite{safra}.
However we are in a more lucky situation as we can tradeoff the lengths
of bottom-level and top-level trees.

To the best of our knowledge the two-level rectilinear Steiner tree
problem has not been considered before despite its practical importance \cite{xiang-etal:2010,xiang-etal:2013}.
It is loosely related to the
hierarchical network design problems \cite{alvarez-etal:2014,
  balakrishnan:1991b,current}
or multi-level facility location problems \cite{baiou,byrka}.
However, those problems are structurally different,
typically  considering problems in graphs, and do not apply to our case.

In \cite{xiang-etal:2010}, ordinary rectilinear Steiner trees were
used to build power efficient fan-in trees, when each input and
the output consists of a single terminal.
In practice designers are also interested in the depth of the constructed
circuit  \cite{xiang-etal:2013}. However, for finding good
power versus depth tradeoffs a better understanding of short solutions
is an essential prerequisite and the aim of our work.

\subsection{Our Contribution}

In Section~\ref{sec:simple-bottom-up} we show even for arbitary metrics that the na{\"i}ve approach
of picking a random terminal from each partition as a connection
point to the top-level tree and building the bottom-level trees and
top-level as separate instances gives a $2\alpha$-factor
approximation, where $\alpha$ is the approximation factor of the used
minimum Steiner tree algorithm.

Then in Section~\ref{sec:arora} and~\ref{sec:predetermined-connection-points} 
we focus on rectilinear instances. In Section~\ref{sec:arora} we show how to lift our instance into
an equivalent $(2+k)$-dimensional rectilinear Steiner tree instance.  If
the number $k$ of partitions is bounded by a constant, we obtain a
PTAS by applying Arora's PTAS for rectilinear Steiner trees
\cite{arora}.

As our main result we improve the approximation guarantee for unbounded $k$ 
from $(2+\epsilon)$ to $1.63$  in Section~\ref{sec:predetermined-connection-points}.
Using spanning tree heuristics this approach turns also into a fast practical algorithm with running time  $\mathcal{O}(n\log n)$ 
and approximation factor $2.37$.

\section{Simple Bottom-Up Construction} 
\label{sec:simple-bottom-up}
A simple bottom-up approach, which works for any metric space, is to compute a Steiner tree $T_i$ for $P_i$
($i=1,\dots, k$).  In each $T_i$ we fix a connection point $q_i \in P_i$
arbitrarly, compute a Steiner tree $T_{top}$ for $\{q_1,
\ldots, q_k\}$, and return $T = (T_{top}, T_1, \ldots, T_k)$.

\begin{theorem}
The simple bottom-up approach is a $2\alpha$-factor approximation algorithm for the minimum two-level Steiner tree problem, if we use
an $\alpha$-factor approximation algorithm for the minimum Steiner tree problem as a subroutine. 
\end{theorem}

\begin{proof}
Let $T$ be the two level Steiner tree computed by the simple bottom-up approach
and let $T^{\star} = (T^{\star}_{top}, T^{\star}_1, \ldots, T^{\star}_k)$ be a minimum two-level Steiner tree. 
%
Let be $q^{\star}_i \in T^{\star}_{top} \cap T^{\star}_i$ the 
connection point of the optimum two-level Steiner tree. 
Since $T^{\star}_i$ is a Steiner tree on $\{q^{\star}_i\} \cup P_i$, we have $dist(q^{\star}_i, q_i) \le l(T^{\star}_i)$.
Thus,
\begin{align*} 
l(T) & = l(T_{top}) + \sum_{i = 1}^{k} l(T_i) ~\le~ \alpha \cdot l(T^{\star}_{top}) + \alpha \sum_{i = 1}^{k} dist(q^{\star}_i, q_i) + \sum_{i = 1}^{k} \alpha \cdot l(T^{\star}_i) \\ 
&\le \alpha \cdot l(T^{\star}_{top}) + 2 \alpha \sum_{i = 1}^{k} l(T^{\star}_i) 
~\le~ 2 \alpha \cdot l(T^{\star}).
\end{align*}
The first inequality follows, since $E(T^{\star}_{top}) \cup \left( \bigcup_{i = 1}^{k} \{q^{\star}_i,q_i\} \right)$ covers a Steiner tree on $\{q_1, \ldots, q_k\}$.
\end{proof}

\begin{figure} [t]
\label{fig:simple}
\begin{minipage}[c]{0.49\textwidth}
\centering
\fbox{
\begin{tikzpicture}[line cap=round,line join=round,x=0.6cm,y=0.3cm]
\clip(-3,-1.5) rectangle (3,1.5);
\fill [color=blue] (2,0) circle (2.5pt);
\fill [color=darkGreen] (-2,0) circle (2.5pt);
\fill [color=blue] (0,0) circle (2.5pt);
\draw [line width=1.2pt,color=blue] (0,0) -- (2,0);
\draw [line width=1.2pt,color=darkGreen] (0,0) -- (-2,0);
\node [color=blue] at (1,0.8) {$T_1$};
\node [color=darkGreen] at (-1,0.8) {$T_2$};
\node [color=red] at (-0.1,-0.9) {$T_{top}$};
\draw [color=red,line width=0.5pt] (0,0) circle (3pt);
\fill [color=blue] (0,0) circle (2.5pt);
\begin{scope}
\clip (0, 0) circle (2.5pt);
\fill[color=darkGreen] (-1, -1) rectangle (0,1);
\end{scope}
\end{tikzpicture}
}
\end{minipage}
\hfill
\begin{minipage}[c]{0.49\textwidth}
\centering
\fbox{
\begin{tikzpicture}[line cap=round,line join=round,x=0.6cm,y=0.3cm]
\clip(-3,-1.5) rectangle (3,1.5);
\draw [line width=1.2pt,color=red] (-2,-0.07) -- (2,-0.07);
\fill [color=blue] (2,0) circle (2.5pt);
\fill [color=darkGreen] (-2,0) circle (2.5pt);
\draw [line width=1.2pt,color=blue] (0,0.07) -- (2,0.07);
\draw [line width=1.2pt,color=darkGreen] (0,0.07) -- (-2,0.07);
\node [color=blue] at (1,0.8) {$T_1$};
\node [color=darkGreen] at (-1,0.8) {$T_2$};
\node [color=red] at (0.75,-0.9) {$T_{top}$};
\fill [color=darkGreen!50!blue] (0,0) circle (2.5pt);
\draw [color=red,line width=0.5pt] (2,0) circle (3pt);
\draw [color=red,line width=0.5pt] (-2,0) circle (3pt);
\node [color=red] at (2.3,0.5) {$q_1$};
\node [color=red] at (-2.3,0.5) {$q_2$};
\fill [color=blue] (0,0) circle (2.5pt);
\begin{scope}
\clip (0, 0) circle (2.5pt);
\fill[color=darkGreen] (-1, -1) rectangle (0,1);
\end{scope}

\end{tikzpicture}
}
\end{minipage}
\caption{A tight example when choosing connection points as arbitary points of $P_i$.} 
\end{figure}
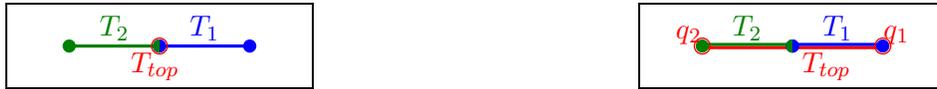

Figure~\ref{fig:simple} shows that the factor $(2+\epsilon)$ is sharp.
For the instance $P_1 = \{ (0,0), (1,0) \}$, $P_2 = \{ (0,0), (-1,0)\}$
a minimum  two-level rectilinear  Steiner tree of length $2$ is shown on the left with $l(T_{top})=0$.
On the right, a bad choice of connection points and minimum Steiner trees $T_{top},T_1$, and $T_3$ yield a total length of $4$.

\section{PTAS for a bounded number of partitions}
\label{sec:arora}
We can reduce the two-level rectilinear Steiner tree problem in the plane to an
ordinary rectilinear Steiner tree problem in a higher dimensional
space, where we can apply Arora's PTAS \cite{arora}.


The idea of the PTAS is to lift every subset $P_1,\dots, P_k$ to an additional dimension.
We assume $k>1$. Otherwise the two-level Steiner tree problem is an ordinary Steiner
tree problem.
Let $P_1,\dots,P_k\subset \Rset^2$ be the subsets of a two-level
Steiner tree instance, we define a Steiner tree instance in
$\Rset^{2+k}$.
The set of terminals $P'$ is comprised as follows.

For each original terminal $x \in P_i \subset \Rset^2$ ($i\in
\{1,\dots,k\}$), we add a terminal $x' := (x,K \cdot e_i) \in
\Rset^{2+k}$, where $e_i\in \Rset^k$ is the unit vector with value one
at the $i$-th coordinate and $K$ is a large constant, e.\ g.\ we could
choose $K$ as  $l(B(P))$.  Now for $x \in P_h$ and $y \in P_i$ the distance of
their high dimensional copies $x',y'\in P'$ is $\lVert x'-y'\rVert_1 =
\lVert x-y\rVert_1 + 2K \lVert e_h-e_i\rVert_1 = \lVert x-y\rVert_1 + 2K
\delta_{h,i}$, where $\delta_{h_i}$ is one if $h = i$ and zero
otherwise.
An example of a  lifted two-level Steiner tree is given in Figure~\ref{fig:lifted-instance}.

\begin{figure}[t] 
\centering 
\begin{tikzpicture}[scale=1]
\pgfmathsetmacro{\cubex}{5}
\pgfmathsetmacro{\cubey}{2.5}
\pgfmathsetmacro{\cubez}{2.5}

\draw [line width=1.2pt,color=red] (-0.3*\cubex,0,-0.6*\cubez) -- (-0.5*\cubex,0,-0.6*\cubez);
\fill [color=red, line width=0.5pt] (-0.3*\cubex,0,-0.6*\cubez)  circle (2.8pt);
\fill [color=red, line width=0.5pt] (-0.5*\cubex,0,-0.6*\cubez)  circle (2.8pt);
\draw [line width=1.2pt,color=red,dash pattern=on 3pt off 3pt] (-0.3*\cubex,0,-0.6*\cubez) -- (-0.3*\cubex,2,-0.6*\cubez);
\draw [line width=1.2pt,color=red,dash pattern=on 3pt off 3pt] (-0.5*\cubex,0,-0.6*\cubez) -- (-0.5*\cubex-3,-2,-0.6*\cubez);

\draw [line width=1.2pt,color=darkGreen] (-0.5*\cubex-3,-2,-0.8*\cubez) -- (-0.5*\cubex-3,-2,-0.3*\cubez) -- (-0.8*\cubex-3,-2,-0.3*\cubez);
\fill [color=darkGreen] (-0.5*\cubex-3,-2,-0.8*\cubez)  circle (2.8pt);
\fill [color=darkGreen] (-0.8*\cubex-3,-2,-0.3*\cubez)  circle (2.8pt);
\node [color=darkGreen] at (-0.6*\cubex-3,-1.8,-0.3*\cubez) { \tiny $T_2$};

\draw [line width=1.2pt,color=blue] (-0.2*\cubex,2,-0.4*\cubez) -- (-0.3*\cubex,2,-0.4*\cubez) -- (-0.3*\cubex,2,-0.8*\cubez);
\fill [color=blue] (-0.2*\cubex,2,-0.4*\cubez)  circle (2.8pt);
\fill [color=blue] (-0.3*\cubex,2,-0.8*\cubez)  circle (2.8pt);
\node [color=blue] at (-0.33*\cubex,1.8,-0.4*\cubez) { \tiny $T_1$};

\fill [color=red, line width=0.5pt] (-0.3*\cubex,2,-0.6*\cubez)  circle (2.8pt);
\fill [color=red, line width=0.5pt] (-0.5*\cubex-3,-2,-0.6*\cubez)  circle (2.8pt);
\node [color=red] at (-0.36*\cubex,-0.2,-0.6*\cubez) {\tiny $T_{top}$};

\draw[black] (0,0,0) -- ++(-\cubex,0,0) -- ++(0,0,-\cubez) -- ++(\cubex,0,0) -- cycle;
\draw[black] (0,2,0) -- ++(-\cubex,0,0) -- ++(0,0,-\cubez) -- ++(\cubex,0,0) -- cycle;
\draw[black] (-3,-2,0) -- ++(-\cubex,0,0) -- ++(0,0,-\cubez) -- ++(\cubex,0,0) -- cycle; 
\node [color=black] at (0.4*\cubex,0,-0.5*\cubez) {\tiny $T_{top}$-space};
\node [color=black] at (0.4*\cubex,2,-0.5*\cubez) {\tiny $T_1$-space};
\node [color=black] at (0.4*\cubex,-2,-0.5*\cubez) {\tiny $T_2$-space};
\end{tikzpicture}

\caption{A flat Steiner Tree in a lifted instance.}
\label{fig:lifted-instance}
\end{figure}
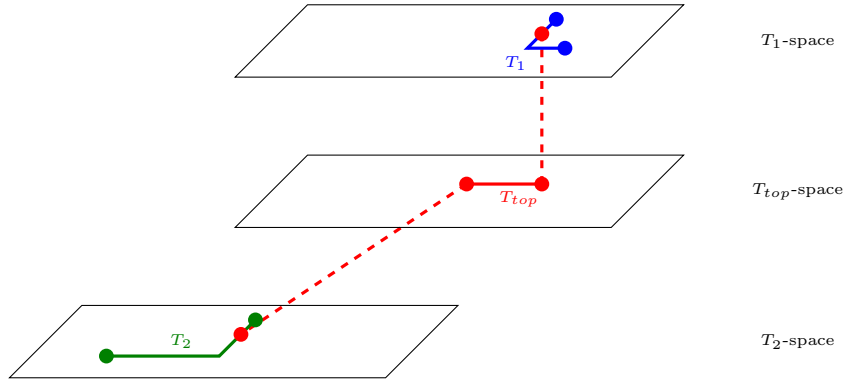

A $(k+2)$-dimensional Steiner tree is called {\it flat} if all Steiner points 
have either the form $(x,0)\in \Rset^{2+k}$ or
$(x,K\cdot e_i)\in \Rset^{2+k}$, where $x\in \Rset^2$ and $e_i\in
\Rset^k$ is a unit vector.
The following Lemma has essentially been proven by Snyder
\cite{snyder}, who shows that an optimum Steiner tree can be found in
the $d$-dimensional Hanan grid \cite{hanan}. We give a short constructive proof for our case.

\begin{lemma} 
  A $(k+2)$-dimensional Steiner Tree $T$ for $P'$ of length $l(T)$ can be transformed in strongly polynomial time into 
  a $(k+2)$-dimensional flat Steiner tree $T'$ of length at most $l(T)$. 
\label{lem:transform-to-flat}
\end{lemma}
\begin{proof}
  We assume that all vertices and segments of $T$ are located within
  the bounding box $B(P') = B(P) \times [0,K]^k$ of its terminals. Otherwise
  we could project $T$ into the box without increasing the length.
  Furthermore we assume that all edges in $E(T)$ are one-directional segments after
  introducing vertices of degree two if necessary.

  Fix a dimension $j \in \{3, \ldots, k+2\}$ and call an edge in
  direction $j$ {\it straight}.  
  Let $X$ be the set of straight segments and let $F$ be the forest
  that arises from $T$ after removing all straight segments.
  
  We construct an undirected  graph $G$ as follows. The vertices of $G$ consist of
  the trees in $F$ plus two extra vertices $s$ and $t$.  We add an
  edge between $s$ and each vertex whose underlying tree is located in
  the hyper plane $H_0 := \{x: x_j = 0\}$. Similarly we add an edge
  between the vertices whose tree is located in $H_j:= \{x: x_j = K\}$
  and $t$. Finally, for each straight edge in $X$ we insert an edge
  between the two vertices representing the corresponding trees in
  $F$.

  We assign all edges in $E(G)$ a unit capacity and compute a minimum
  $s$-$t$-cut $\delta(S)$, where $S \subset V(G)$ with $s \in S$. 
  $T^j$ is assembled from $T$
  by projecting all subtrees in $F$ whose corresponding vertex is in $S$
  into $H_0$ and all other trees into $H_j$.  
  To get a valid Steiner tree, we eliminate potentially
  arising cycles by removing edges arbitrarily.
  By the projection, straight edges in the cut $\delta(S)$ will grow to a length of $K$
  and straight edges outside will shorten to length zero  and can be eliminated.

  By Menger's theorem \cite{menger} there are $|\delta(S)|$
  edge-disjoint paths between $s$ and $t$.  Each path represents a set
  straight edges of total length at least $K$. Thus $l(X) \ge  |\delta(S)| \cdot K$ and
  $$ l(T^j) =  |\delta(S)| \cdot K + l (F) \le l(X) + l(F) = l(T).$$
  
  Frist, we apply this transformation to all dimensions $j \in \{3, \ldots
  k+2\}$ to obtain a Steiner tree with Steiner points $x$ of value either $x_j = 0$ or $x_j = K$ in the dimensions $j \in \{3, \ldots
  k+2\}$ without increasing the length of $T$. Second, we remap all Steiner points $x$ with at least two
  coordinates $j,i\in\{3,\dots,k+2\}$ of value $x_i = x_j = K$ to
  $\Rset^2\times \{0\}$ by setting all $x_j$ $j\in\{3,\dots,k+2\}$ to
  zero, we obtain a flat Steiner tree $T'$ of lenght at most $l(T)$, since we are in the rectilinear case.
  
  The running time is dominated by the $k$ minimum cut
  computations that can be done in a digraph, where each edge in
  $E(G)$ is represented by two oppositely directed edges, in strongly
  polynomial time \cite{edmonds+karp:1972}.
\end{proof}

Next we show that we can assume that a tree $T$ for $P'$ has at most one edge
in direction $j\in \{3,\dots,k\}$.

\begin{lemma} 
  If $K \ge l(B(P))$, a
  $(k+2)$-dimensional flat Steiner tree $T$ for $P'$ of length $l(T)$ can be transformed in strongly polynomial time into 
  a $(k+2)$-dimensional flat Steiner tree $T'$ for $P'$ of length at most $l(T)$ so that $T'$ contains at most one edge 
  in each lifting  direction $j \in \{3, \ldots k+2\}$. 
\label{lem:transform-to-two-level}
\end{lemma}
\begin{proof}
  Assume that there is a direction $j\in\{3,\dots,k+2\}$, s.t. $T$
  contains two edges $e, e'$ in direction $j$. Removing $e'$ splits
  $T$ into two components which we can re-connect by an edge in either
  $H_0 := \{x: x_j = 0\}$ or $H_j := \{x: x_j = K\}$ between an
  endpoint of $e$ and $e'$. The new edge has length at most $l(B(P))
  \le K$, the length of $e'$.
\end{proof}

The proof also shows that if $K> l(B(P))$ and $T$ has minimum length, 
it contains at most one edge in each lifted direction $j \in \{3, \ldots k+2\}$.
The following lemma shows the equivalence between the original
two-level rectilinear  Steiner tree problem in the plane and the lifted
regular rectilinear Steiner tree problem.

\begin{lemma}
  If $k > 1$, a two-level rectilinear Steiner tree $T$ for $P_1,\dots,P_k$ of length
  $l(T)$ can be transformed into a $(k+2)$ dimensional Steiner tree $T'$
  for $P'$ of length at most $l(T)+kK$ and vice versa.
\label{lem:kd-reduction}
\end{lemma}
\begin{proof}
  W.l.o.g. we assume that the center of the coordinate system coincides with the
  bounding box center of $P$.
  Let $T = (T_{top}, T_1, \ldots, T_k)$ be a two-level Steiner tree of
  length $l(T)$.
  We embed the vertices $x \in V(T_{top})\subset \Rset^2$ into $\Rset^2
  \times \{0\}^k$ as $(x,0)$, where $0$ is a k-dimensional zero
  vector.  
  For $T_i$ ($i=\{1,\dots,k\}$, we embed $x \in V(T_i)\subset \Rset^2$
  as $(x,K\cdot e_i)\in \Rset^{2+k}$, where $e_i\in \Rset^k$ is again a unit vector 
  with value one in its $i$-th coordinate.

  To connect the top-level and its subtrees, we pick for each subtree
  $i\in\{1,\dots,k\}$ a connection point $q_i\in T_{top}\cap T_i$ and
  connect the lifted components by a new edge $\{(q_i,0),(q_i,K\cdot
  e_i)\}$ of length $K$. Clearly, the length of the lifted tree is 
  $l(T) + kK$.

  Now let $T'$ be a rectilinear Steiner tree for $P'$
  of length $l(T')$. By applying Lemma \ref{lem:transform-to-flat} we
  obtain a $(k+2)$-dimensional flat Steiner tree. 
  Applying Lemma \ref{lem:transform-to-two-level}, we can further assume that
  for each $j\in\{3,\dots,k+2\}$, $T'$ contains   at most one edge in direction $j$.

  Removing all $k$ edges of length $K$, $T'$ is decomposed into $k+1$
  subtrees.  Projecting these onto the first two coordinates we obtain
  a feasible two-level Steiner tree of length at most $l(T')-kK$.
\end{proof}

\begin{theorem}
  For bounded $k$ there is a PTAS for the two-level rectilinear  Steiner tree problem.
\end{theorem}
\begin{proof}
Choose $K=l(B(P))$ and for $\epsilon > 0$ set $\epsilon' := \frac{1}{k+1} \epsilon$. 
Then compute an $(1+\epsilon')$-approximate  $(k+2)$-dimensional 
Steiner tree $T'$ for the lifted terminal set $P'$ with Aroras PTAS \cite{arora} that has a polynomial running in bounded dimension. 
Then we apply Lemma~\ref{lem:kd-reduction} to obtain a
two-level Steiner tree $T= (T_{top},T_1,\dots,T_k)$ for $P_1,\dots, P_k$ with length at most $l(T') - kK$.
Let $T'^{\star}$ and $T^{\star}$ be optimum Steiner trees for $P'$ and $P$. Since $l(T^{\star}) \ge K$ the length of $T$ is
\begin{equation}
\begin{array}{rl}
l(T) & \le l(T') -  kK  \le (1+\epsilon') l(T'^{\star}) - kK \le (1+\epsilon') l(T^{\star}) + (1+\epsilon')kK -kK \\
     &  \le (1+(k+1)\epsilon') l(T^{\star}) = (1+\epsilon) l(T^{\star}).
\end{array}
\end{equation}
\end{proof}

Note that $\lceil \frac{k}{2} \rceil + 2$ dimensions would also be sufficient
to achieve this result by lifting two partitions into one extra dimension (one partition in positive direction and
the other in negative direction). 
This lifting method does not work using Steiner tree approximation algorithms with  a constant factor $\alpha$, because
we get an additional additive error of $\alpha \cdot k\cdot K$, which in total is  essentially not better than the $2\alpha$-factor
approximation by the simple algorithm in Section~\ref{sec:simple-bottom-up}.

Our construction transforms an optimum $(k+2)$-dimensional Steiner tree for $P'$ in the $(k+2)$-dimensional Hanan grid, 
which always exist due to Snyder  \cite{snyder}, into a two-level rectilinear Steiner tree in the 2-dimensional  Hanan grid defined by $P$.
\begin{cor}
 Given a set of terminals $P = (P_1, \ldots, P_k)$, there is always an optimum two-level rectilinear Steiner tree for $P$
 in its Hanan grid.
 \label{cor:contained-in-Hanan-grid}
\end{cor}
This also proves the following simple fact.
\begin{cor}
  \label{cor:contained-in-NP}
  The  two-level rectilinear Steiner tree problem is in \NP.
\end{cor}

\section{Predetermined Connection Points} 
\label{sec:predetermined-connection-points}

In all  algorithms of this section we predetermine a connection point $q_i$ for
each set $P_i$ $(i=1,\dots,k$) and then call a Steiner tree
approximation algorithm for $\{q_1,\dots,q_k\}$ to get $T_{top}$ and
$P_i\cup \{q_i\}$ to get $T_i$ ($i=1\dots,k$).
We use the fact that we consider rectilinear instances to obtain
better approximation factors than in Section~\ref{sec:simple-bottom-up}.

\subsection{Bounding Box Center} 
\label{sec:bounding-box-center}
A natural approach is to choose each connection point
$q_i$ ($i\in \{1,\dots,k\}$) as the  center of the bounding box $B(P_i)$.

\begin{theorem}
  Using bounding box centers as connection points, we get a $1.75\alpha$-factor approximation algorithm for the two-level rectilinear Steiner tree problem,
  when using an $\alpha$-factor approximation algorithm for rectilinear Steiner trees as a subroutine. 
\end{theorem}
\begin{proof}
  Let $T$ be the resulting two level Steiner tree and let $T^{\star} = \{T^{\star}_{top}, T^{\star}_1, \ldots, T^{\star}_k\}$ be a
  minimum two-level Steiner tree. 
  We choose $T^{\star}_{top}$ under all minimum two-level Steiner trees as long as possible so that we can choose connection points  $q^{\star}_i \in T^{\star}_{top}
  \cap T^{\star}_i \cap B(P_i)$ for all $i\in\{1,\dots,k\}$.

  Now, consider a partition set $P_i$ ($i\in\{1,\dots,k\})$.  We assume w.l.o.g. that the
  horizontal length of $B(P_i)$ is no less than the vertical length.
  Let $q'_i$ be a point in the intersection of $T^{\star}_i$ and the vertical line through $q_i$. 
  Note that $\mnorm{q'_i - q_i} \le \frac{1}{4}
  l(B(P_i))$ by the shape of $B(P_i)$.

  Now  for each $i \in \{1,\dots,k\}$, $E\left(T^{\star}_{top}\right)\cup \left(\bigcup_{i = 1}^{k} \{q^{\star}_i,q_i\}\right)$ covers a Steiner tree for $\{q_1, \ldots, q_k\}$, and $E(T^{\star}_i) \cup \{q_i,q'_i\}$ covers a Steiner tree for $P_i \cup \{q_i\}$. Thus,
\begin{align*} 
l(T) & = l(T_{top}) + \sum_{i = 1}^{k} l(T_i) \\
&\le \alpha \cdot l(T^{\star}_{top}) + \alpha \sum_{i = 1}^{k} (\lVert q^{\star}_i - q_i\rVert_1) + \sum_{i = 1}^{k} l(T_i) \\ 
&\le \alpha \cdot l(T^{\star}_{top}) + \alpha \sum_{i = 1}^{k} (\lVert q^{\star}_i - q_i\rVert_1) 
 + \alpha \sum_{i = 1}^{k} l(T^{\star}_i) + \alpha \sum_{i = 1}^{k} (\lVert q'_i - q_i\rVert_1) \\ 
&\le \alpha \cdot l(T^{\star}_{top}) + \frac{7}{4} \alpha \sum_{i = 1}^{k} l(T^{\star}_i) ~\le~ \frac{7}{4} \alpha \cdot (l(T^{\star}_{top}+ \sum_{i = 1}^{k} l(T^{\star}_i)) ~\le~ 1.75 \alpha \cdot l(T^{\star}).
\end{align*}
\end{proof}


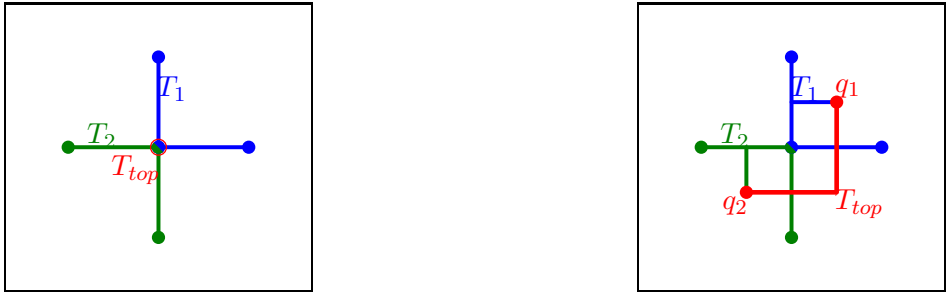
\begin{figure} [!t]
\begin{minipage}[c]{0.49\textwidth}
\centering
\fbox{
\begin{tikzpicture}[line cap=round,line join=round,x=0.6cm,y=0.6cm]
\clip(-3,-3) rectangle (3,3);
\fill [color=blue] (0,2) circle (2.5pt);
\fill [color=blue] (2,0) circle (2.5pt);
\fill [color=darkGreen] (-2,0) circle (2.5pt);
\fill [color=darkGreen] (0,-2) circle (2.5pt);
\fill [color=darkGreen!50!blue] (0,0) circle (2.5pt);
\draw [line width=1.5pt,color=blue] (0,0) -- (0,2);
\draw [line width=1.5pt,color=blue] (0,0) -- (2,0);
\draw [line width=1.5pt,color=darkGreen] (0,0) -- (0,-2);
\draw [line width=1.5pt,color=darkGreen] (0,0) -- (-2,0);
\node [color=blue] at (0.3,1.3) {$T_1$};
\node [color=darkGreen] at (-1.25,0.26) {$T_2$};
\node [color=red] at (-0.5,-0.5) {$T_{top}$};
\draw [color=red,line width=0.5pt] (0,0) circle (3pt);
\fill [color=blue] (0,0) circle (2.5pt);
\begin{scope}
\clip (0, 0) circle (2.5pt);
\fill[color=darkGreen] (-1, 1) -- (1,-1) -- (-1,-1);
\end{scope}
\end{tikzpicture}
}
\end{minipage}
\hfill
\begin{minipage}[c]{0.49\textwidth}
\centering
\fbox{
\begin{tikzpicture}[line cap=round,line join=round,x=0.6cm,y=0.6cm]
\clip(-3,-3) rectangle (3,3);
\fill [color=blue] (0,2) circle (2.5pt);
\fill [color=blue] (2,0) circle (2.5pt);
\fill [color=darkGreen] (-2,0) circle (2.5pt);
\fill [color=darkGreen] (0,-2) circle (2.5pt);
\draw [line width=1.5pt,color=blue] (0,0) -- (0,2);
\draw [line width=1.5pt,color=blue] (0,1) -- (1,1);
\draw [line width=1.5pt,color=blue] (0,0) -- (2,0);
\draw [line width=1.5pt,color=darkGreen] (0,0) -- (0,-2);
\draw [line width=1.5pt,color=darkGreen] (0,0) -- (-2,0);
\draw [line width=1.5pt,color=darkGreen] (-1,0) -- (-1,-1);
\node [color=blue] at (0.3,1.3) {$T_1$};
\node [color=red] at (1.5,-1.25) {$T_{top}$};
\fill [color=red] (1,1) circle (2.5pt);
\fill [color=red] (-1,-1) circle (2.5pt);
\node [color=red] at (1.25,1.3) {$q_1$};
\node [color=red] at (-1.25,-1.3) {$q_2$};
\node [color=darkGreen] at (-1.25,0.26) {$T_2$};
\draw [line width=1.75pt,color=red] (-1,-1) -| (1,1);
\fill [color=blue] (0,0) circle (2.5pt);
\begin{scope}
\clip (0, 0) circle (2.5pt);
\fill[color=darkGreen] (-1, 1) -- (1,-1) -- (-1,-1);
\end{scope}
\end{tikzpicture}
}
\end{minipage}
\caption{A tight example when choosing connection points as bounding box centers.} 
\label{fig:sharp-example-center-point}
\end{figure}

Figure~\ref{fig:sharp-example-center-point} shows that the factor $(\frac{7}{4}+\epsilon)$  is sharp.
For the instance $P_1 = \{ (0,0), (1,0), (0,1) \}$, $P_2 = \{ (0,0), (-1,0), (0,-1) \}$
a minimum two-level rectilinear  Steiner tree of length $4$ is shown on the left with $l(T_{top})=0$.
On the right,  with bounding box centers as connection points and  shortest Steiner trees the
  total length is  $7$.

\subsection{Adjusted Bounding Box Center}
We can improve the approximation factor by a more careful choice of
the connection point.
For a set $P_i$ ($i\in \{1,\dots,k\}$), we call the coordinate system with
origin in the central point of its bounding box the {\it coordinate
  system of $P_i$}.
If no terminal in a subtree instance $P_i$ is
located in the lower left quadrant of the bounding box w.r.t. its
coordinate system, it appears reasonable to shift the connection point
to the upper right towards the actual terminals, e.\ g.\ one would  move $q_2$ towards the upper right in
Figure~\ref{fig:sharp-example-center-point}.

Otherwise, if a set $P_i$ of terminals contains an element in each
quadrant of its bounding box $B(P_i)$, we call the bounding box
$B(P_i)$ {\it complete}, as  the left example in Figure~\ref{fig:adjusted-box-center} shows.
For subtrees with a complete bounding box we choose the connection point
to the top-level tree as the central point of the bounding box as in
Section~\ref{sec:bounding-box-center}.

\begin{figure} [!b]
\begin{minipage}[c]{0.49\textwidth}
\centering
\fbox{
\begin{tikzpicture}[line cap=round,line join=round,x=0.6cm,y=0.6cm]
\draw[<->,color=black] (-6,0) -- (6,0);
\foreach \x in {-5,-4,-3,-2,-1,,1,2,3,4,5}
\draw[shift={(\x,0)},color=black] (0pt,2pt) -- (0pt,-2pt) node[below] {\footnotesize $\x$};
\draw[<->,color=black] (0,-3) -- (0,3);
\foreach \y in {-2,-1,,1,2}
\draw[shift={(0,\y)},color=black] (2pt,0pt) -- (-2pt,0pt) node[left] {\footnotesize $\y$};
\clip(-6,-3) rectangle (6,3);
\draw[color=black] (-0.24,-0.23) node {\footnotesize $0$};
\draw [line width=1.2pt,color=blue] (-5,2)-- (5,2);
\draw [line width=1.2pt,color=blue] (-5,2)-- (-5,-2);
\draw [line width=1.2pt,color=blue] (-5,-2)-- (5,-2);
\draw [line width=1.2pt,color=blue] (5,-2)-- (5,2);
\draw[color=blue] (1.24,2.24) node {$B(P_i)$};
\fill [color=red] (0,0) circle (2.5pt);
\draw[color=red] (0.24,0.24) node {$q_i$};
\fill [color=blue] (5,1) circle (2.5pt);
\fill [color=blue] (3,2) circle (2.5pt);
\fill [color=blue] (4,-2) circle (2.5pt);
\fill [color=blue] (-5,0.5) circle (2.5pt);
\fill [color=blue] (-1,-1.5) circle (2.5pt);
\fill [color=blue] (-2,-1) circle (2.5pt);
\fill [color=blue] (2,1) circle (2.5pt);
\fill [color=blue] (3,-0.5) circle (2.5pt);
\fill [color=blue] (-3,1) circle (2.5pt);
\end{tikzpicture}
}
\end{minipage}
\hfill
\begin{minipage}[c]{0.49\textwidth}
\centering
\fbox{
\begin{tikzpicture}[line cap=round,line join=round,x=0.6cm,y=0.6cm]
\draw[<->,color=black] (-6,0) -- (6,0);
\foreach \x in {-5,-4,-3,-2,-1,,1,2,3,4,5}
\draw[shift={(\x,0)},color=black] (0pt,2pt) -- (0pt,-2pt) node[below] {\footnotesize $\x$};
\draw[<->,color=black] (0,-3) -- (0,3);
\foreach \y in {-2,-1,,1,2}
\draw[shift={(0,\y)},color=black] (2pt,0pt) -- (-2pt,0pt) node[left] {\footnotesize $\y$};
\clip(-6,-3) rectangle (6,3);
\draw[color=black] (-0.24,-0.23) node {\footnotesize $0$};
\draw [line width=1.2pt,color=blue] (-5,2)-- (5,2);
\draw [line width=1.2pt,color=blue] (-5,2)-- (-5,-2);
\draw [line width=1.2pt,color=blue] (-5,-2)-- (5,-2);
\draw [line width=1.2pt,color=blue] (5,-2)-- (5,2);
\draw[color=blue] (1.24,2.24) node {$B(P_i)$};
\fill [color=red] (1,1) circle (2.5pt);
\draw[color=red] (1.34,0.77) node {$q_i$};
\draw [line width=1.2pt,color=red,dash pattern=on 3pt off 3pt] (-5,1) -- (1,1);
\draw [line width=1.2pt,color=red,dash pattern=on 3pt off 3pt] (1,1) -- (1,-2);
\draw [line width=1.2pt,color=orange,dash pattern=on 3pt off 3pt] (1.5,2) -- (1.5,1.5);
\draw [line width=1.2pt,color=orange,dash pattern=on 3pt off 3pt] (5,1.5) -- (1.5,1.5);
\fill [color=blue] (5,-1) circle (2.5pt);
\fill [color=blue] (3,-2) circle (2.5pt);
\fill [color=blue] (4,-2) circle (2.5pt);
\fill [color=blue] (-5,1.5) circle (2.5pt);
\fill [color=blue] (-1,2) circle (2.5pt);
\fill [color=blue] (-2,1) circle (2.5pt);
\fill [color=blue] (2,-1) circle (2.5pt);
\fill [color=blue] (3,0.5) circle (2.5pt);
\fill [color=blue] (4,1.5) circle (2.5pt);
\end{tikzpicture}
}
\end{minipage}
\caption{On the left an example for a complete bounding box with $q_i$ chosen at the center.
On the right an example for an incomplete bounding box.
$(t_i^1,t_i^1)$ is the bend  of the red line, here defining $q_i$, and $(t_i^2,t_i^2)$  the bend on the orange line.}
\label{fig:adjusted-box-center}
\end{figure}
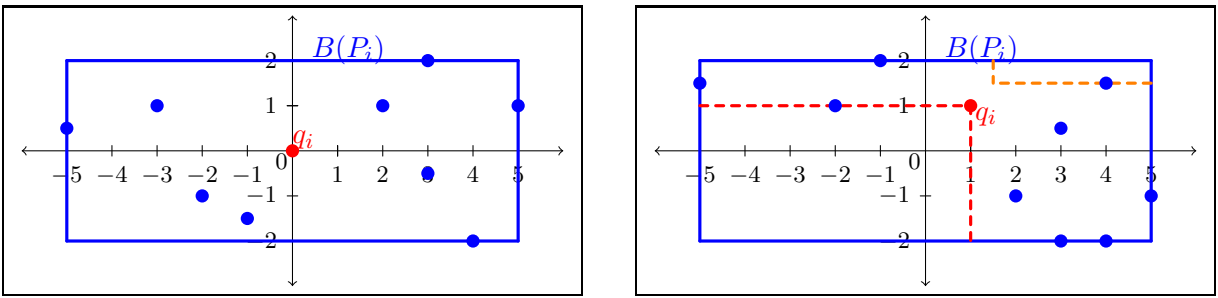

For sets $P_i$ with an incomplete bounding box $B(P_i)$ we assume 
w.l.o.g. (after reflection and rotation) that there is no point
of $P_i$ in the lower-left quadrant of the coordinate system for $P_i$
and that the horizontal length is not less than the
vertical length  of the bounding box $B({P_i})$.  
We then define 
\begin{align} 
t_i^1 &= \max\{s \in [0, t_i^{\max}]\;:\; P \cap \{(x,y): x<s,y<s\}= \emptyset\}, \nonumber \\
t_i^2 &= \min\{s \in [0, t_i^{\max}]\;:\; P \cap \{(x,y): x>s,y>s\}= \emptyset\}, \text{ and} \label{defl} \\
t_i &= \min \{ t_i^1 , t_i^2, \beta t_i^{\max} \},  \nonumber
\end{align}
where $\beta \in [0,1]$ is a parameter we will choose later and $t_i^{\max}$ is half the vertical length of $B(P_i)$,
which by our assumption is also the (vertical) distance from the center to the boundary of $B(P_i)$.

For subtrees with incomplete bounding box we choose the connection
point $q_i = ({q_i}_x, {q_i}_y)$ to the top-level tree as the central point of the bounding
box plus the vector $(t_i, t_i)$ (see also
Figure~\ref{fig:adjusted-box-center}).
As before,  we compute Steiner trees $T_{top}$ for $\{q_1, \ldots, q_k\}$ and $T_i$ for $P_i \cup \{q_i\}$ ($i\in \{1,\dots,k\}$).

Moreover for $i\in \{1,\dots,k\}$ the trees $T_i$ are refined as
follows.  We compute a second Steiner tree $T'_i$ for $P_i \cup
\{q_i\}$ as follows: We compute a Steiner tree $T'_i$ for
$P_i$. Thereby, we embed maximal paths in $T'_i$ containing only
 Steiner vertices with degree two so that each such path has minimum distance to
$q_i$ while preserving its length.  We then add an edge from $q_i$ to
$a_i$, where $a_i$ is a point in $T'_i$ minimizing the distance to
$q_i$.
Finally, if $l(T'_i)\le l(T_i)$, we  replace $T_i$ by $T'_i$.

Let $T$ be the two-level
Steiner tree computed as described and
let $T^{\star} = \{T^{\star}_{top}, T^{\star}_1, \ldots, T^{\star}_k\}$ be
an minimum two-level Steiner tree. Again we choose $T^{\star}_{top}$
under all minimum two-level Steiner trees as large as possible so that
there is a connection point $q^{\star}_i = ({q^{\star}_i}_x,{q^{\star}_i}_y) \in T^{\star}_{top} \cap
T^{\star}_i \subseteq B(P_i)$.

We start with two lemmas  bounding the distance $\mnorm{q_i-a_i}$
between the connection point $q_i$ and a nearest neighbor $a_i$  in $T'_i$ by the length $l(T_i^{\star})$ of the subtree
in an optimum solution.

\begin{lemma} \label{lemma1}
Let $i\in \{1,\dots,k\}$ and $T'_i, a_i$ be constructed as above.
If $B(P_i)$ is complete, then $l(B(P_i)) + \lVert q_i - a_i \rVert_1 \le l(T^{\star}_i)$.
\end{lemma}
\begin{proof}

Define $h :=  \lVert q_i - a_i \rVert_1$. Since $B(P_i)$ is complete, $T'_i$ intersects at least three of the four axes of the coordinate system to $P_i$. We assume w.l.o.g. that $T'_i$ intersects the left, upper and right axis.
(We did not rotate complete boxes before but only incomplete ones.)

By the choice of $T'_i$ and $h$ there exist (see also Figure~\ref{fig:lemma1})
\begin{align*}
p &\in  \{(x,y) \in P_i : x \le -h, y \le 0\},\\
p'& \in \{(x,y) \in P_i : x \ge h, y \le 0\},\\
u &\in  \{(x,y) \in P_i : x < 0, y \ge h\},\\
v_z &\in V_z :=  \{(x,y) \in P_i: x \ge z, y \ge h-z \} \text{ for all } z \in [0, h].
\end{align*}

Let $v_h \in V_h$. If the unique $T^{\star}_i$-paths from $p$ to $u$ and from $p'$ to $v_h$ intersect, then
$T^{\star}_i$ connects  the lines $\{(x,y) : x = 0\}$ and $\{(x,y) : x = h\}$ twice, and therefore $l(B(P_i)) + h \le l(T^{\star}_i)$.

Otherwise, we can choose a minimum $z \ge 0$  such that there is a $v \in V_z$ and the unique $T^{\star}_i$-paths from $p$ to $u$ and from $p'$ to $v$ are disjoint.
The lines $\{(x,y) : y = 0\}$ and $\{(x,y) : y = h - z\}$ are connected twice in $T^{\star}_i$. Therefore we get
$l(B(P_i)) + h - z \le l(T^{\star}_i)$.

If $z=0$ we are done. Otherwise, if our statement is false there is an $0 < \epsilon \le z$ such that $\epsilon = l(B(P_i)) + h - l(T^{\star}_i)$ and a
$w \in V_{z-\frac{\epsilon}{2}}$.
Since the unique $T^{\star}_i$-paths from $p$ to $u$ and from $p'$ to $w$ are not disjoint,  $T^{\star}_i$ connects the lines  $\{(x,y) : x = 0\}$ and $\{(x,y) : x = z - \frac{\epsilon}{2}\}$  twice. 
Therefore we get the contradiction

\begin{equation*}
\epsilon = l(B(P_i)) + h - l(T^{\star}_i) \le l(B(P_i)) + h - l(B(P_i)) - h + z - z + \frac{\epsilon}{2} = \frac{\epsilon}{2}.
\end{equation*}

\begin{figure} [!t]
\centering 
\fbox{
\begin{tikzpicture}[line cap=round,line join=round,x=0.66cm,y=0.66cm]
\draw[<->,color=black] (-6,0) -- (6,0);
\foreach \x in {-5,-4,-3,-2,-1,,1,2,3,4,5}
\draw[shift={(\x,0)},color=black] (0pt,2pt) -- (0pt,-2pt) node[below] {\footnotesize $\x$};
\draw[<->,color=black] (0,-3) -- (0,3);
\foreach \y in {-2,-1,,1,2}
\draw[shift={(0,\y)},color=black] (2pt,0pt) -- (-2pt,0pt) node[left] {\footnotesize $\y$};
\clip(-6,-3) rectangle (6,3);
\draw[color=black] (-0.24,-0.23) node {\footnotesize $0$};
\draw [line width=1.2pt,color=blue] (-5,2)-- (5,2);
\draw [line width=1.2pt,color=blue] (-5,2)-- (-5,-2);
\draw [line width=1.2pt,color=blue] (-5,-2)-- (5,-2);
\draw [line width=1.2pt,color=blue] (5,-2)-- (5,2);
\draw[color=blue] (1.24,2.24) node {$B(P_i)$};
\draw [line width=1.2pt,color=darkGreen] (-1.5,0)-- (0,1.5);
\draw [line width=1.2pt,color=darkGreen] (1.5,0)-- (0,1.5);
\draw [line width=1.2pt,color=darkGreen] (1.5,0)-- (0,-1.5);
\draw [line width=1.2pt,color=darkGreen] (0,-1.5)-- (-1.5,0);
\fill [color=blue] (5,-1) circle (2.5pt);
\fill [color=blue] (-4,-2) circle (2.5pt);
\fill [color=blue] (4,2) circle (2.5pt);
\fill [color=blue] (-1,1.5) circle (2.5pt);
\draw[color=blue] (-1.24,1.73) node {$u$};
\fill [color=blue] (0.5,1.5) circle (2.5pt);
\draw[color=blue] (0.26,1.73) node {$v$};
\fill [color=blue] (2.5,1) circle (2.5pt);
\draw[color=blue] (2.26,1.23) node {$w$};
\fill [color=blue] (-2,-0.5) circle (2.5pt);
\draw[color=blue] (-2.24,-0.26) node {$p$};
\fill [color=blue] (-2,1) circle (2.5pt);
\fill [color=blue] (2,-1) circle (2.5pt);
\draw[color=blue] (2.24,-1.3) node {$p'$};
\fill [color=blue] (-5,-1) circle (2.5pt);
\fill [color=red] (0,0) circle (2.5pt);
\draw[color=red] (0.24,-0.24) node {$q_i$};
\draw [line width=1pt,color=blue,dash pattern=on 3pt off 3pt] (4,1)-- (0.5,1);
\draw [line width=1pt,color=blue,dash pattern=on 3pt off 3pt] (2,1)-- (2,-1);
\draw [line width=1pt,color=blue,dash pattern=on 3pt off 3pt] (2,-1)-- (5,-1);
\draw [line width=1pt,color=blue,dash pattern=on 3pt off 3pt] (0.5,1)-- (0.5,1.5);
\draw [line width=1pt,color=blue,dash pattern=on 3pt off 3pt] (-1,1.5)-- (0.5,1.5);
\draw [line width=1pt,color=blue,dash pattern=on 3pt off 3pt] (-2,1)-- (-1,1);
\draw [line width=1pt,color=blue,dash pattern=on 3pt off 3pt] (-2,-1)-- (-2,1);
\draw [line width=1pt,color=blue,dash pattern=on 3pt off 3pt] (-5,-1)-- (-2,-1);
\draw [line width=1pt,color=blue,dash pattern=on 3pt off 3pt] (-4,-1)-- (-4,-2);
\draw [line width=1pt,color=blue,dash pattern=on 3pt off 3pt] (4,1)-- (4,2);
\draw [line width=1pt,color=blue,dash pattern=on 3pt off 3pt] (-1,1)-- (-1,1.5);
\end{tikzpicture}
}
\caption{An example of the situation in the proof of Lemma~\ref{lemma1}. The green
diamond is the $l_1$-circle with radius $h$ around $q_i$.}
\label{fig:lemma1}
\end{figure}
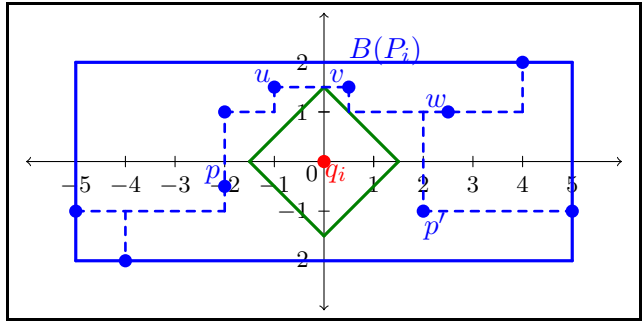
\end{proof}

\begin{lemma} \label{lemma2}
Let $i\in \{1,\dots,k\}$ and $T'_i, a_i$ be constructed as above.
If $B(P_i)$ is incomplete and $t_i < \beta t_i^{\max}$, then
$l(B(P_i)) + \lVert q_i - a_i \rVert_1 \le l(T^{\star}_i)$.
\end{lemma}
\begin{proof}
If $t^2_i \le t^1_i$ all $(x,y) \in P_i$ are either $x \le t^2_i, y \ge t^1_i$ or $x \ge t^1_i, y \le t^2_i$.
By the construction of $T'_i$ we have $a_i = q_i$ and we are done.
Otherwise, there is a point $p = (p_x, p_y) \in P_i$ with either $p_x = t_i^1$ and $p_y \le t_i^1$ or $p_x \le t_i^1$ and $p_y = t_i^1$. 
W.l.o.g. $p_x \le t_i^1$ and $p_y = t_i^1$. Let be $p' \in P_i$ a point in the lower border of the bounding box. 
The $T'_i$-path from $p$ to $p'$ intersects the left, upper and right halfline starting in $q_i$.
We are now in the same setting as in the proof of Lemma~\ref{lemma1}. We obtain analougsly
$l(B(P_i)) + \lVert q_i - a_i \rVert_1 \le l(T^{\star}_i)$. 

\end{proof}

\begin{lemma} \label{lemma6}
If $B(P_i)$ is incomplete,  ${q^{\star}_i}_x <{q_i}_x$, and ${q^{\star}_i}_y <{q_i}_y$, then
$$\lVert q^{\star}_i - q_i\rVert_1 \le \frac{1}{2} l(T^{\star}_i) + \left(\frac{3}{2}\beta -\frac{1}{2}\right)  t_i^{\max}.$$
\end{lemma}
\begin{proof}
Since $B(P_i)$ is incomplete,  $q^{\star}_i\in B(P_i)$, and by the choice of $t_i$ and $q_i$ in (\ref{defl}), there are points 
\[u \in \{(x,y) \in P_i :x \le {q^{\star}_i}_x, y \ge {q_i}_y\},\] 
\[v \in \{(x,y) \in P_i :x \ge {q_i}_x, y \le {q^{\star}_i}_y\},\]
\[w \in \{(x,y) \in P_i :x \ge {q_i}_x, y \ge {q_i}_y\}.\]

By the positions of the four points $q^{\star}_i$, $u$, $v$, and $w$, $T^{\star}_i$ connects either the lines $\{(x,y) : x = {q_i}_x\}$ and $\{(x,y) \in \Rset^2 : x = {q^{\star}_i}_x\}$ or the lines
$\{(x,y) \in \Rset^2: y = {q_i}_y\}$ and $\{(x,y) : y = {q^{\star}_i}_y\}$ twice. Therefore we have 
$\min \{|{q_i}_x - {q^{\star}_i}_x|, |{q_i}_y - {q^{\star}_i}_y|\} + l(B(P_i)) \le l(T^{\star}_i)$ and
\begin{align*}
\lVert q^{\star}_i - q_i\rVert_1 
& =  |{q^{\star}_i}_x - {q_i}_x| + |{q^{\star}_i}_y - {q_i}_y| \\
&\le \frac{1}{2} \left(l(T^{\star}_i) - l(B(P_i))\right) \hspace{-0.05cm}+\hspace{-0.05cm}\frac{1}{2} \min \{|{q_i}_x - {q^{\star}_i}_x|, |{q_i}_y - {q^{\star}_i}_y|\}\hspace{-0.05cm}+\hspace{-0.05cm} \max \{|{q_i}_x - {q^{\star}_i}_x|, |{q_i}_y - {q^{\star}_i}_y|\} \\
&\le \frac{1}{2} l(T^{\star}_i) - \frac{1}{2} l(B(P_i)) + \frac{1}{2} \left(t_i + t_i^{\max}\right) + t_i + \frac{1}{2} l_{hor}(B(P_i))\\
&\le \frac{1}{2} l(T^{\star}_i) + \frac {3}{2} t_i - \frac{1}{2} t_i^{\max} \le \frac{1}{2} l(T^{\star}_i) + \left(\frac{3}{2}\beta -\frac{1}{2}\right) t_i^{\max},
\end{align*}
where $l_{hor}(B(P_i))$ is the horizontal length of $B(P_i)$ and we first use $2 t^{\max}_i \le l_{hor}(B(P_i))$ to dissolve minimum and maximum, and then $l(B(P_i)) = 2 t^{\max}_i+ l_{hor}(B(P_i))$ and $t_i \le \beta t_i^{\max}$.
\end{proof}

\begin{lemma} \label{lemma7}
If $B(P_i)$ is incomplete and if $q^{\star}_i \in \{(x,y) \in B(P_i) : x \ge {q_i}_x \text{ or }  y \ge {q_i}_y\}$, then
$\lVert q^{\star}_i - q_i\rVert_1 \le \frac{1}{2} l(B(P_i))$.
\end{lemma}
\begin{proof}
Since $q_i \in \{(x,y)\in B(P_i) : x = y\}$ we have $\mnorm{q^{\star}_i - q_i} \le \frac{1}{2} l(B(P_i))$.
\end{proof}

We are now able to bound the length of each subtree $T_i$ by the length of its corresponding tree $T^{\star}_i$
in the optimum solution.

\begin{lemma} \label{lemma8}
Let $1 \le \alpha \le 1.5$  be an approximation factor for the rectilinear Steiner tree problem. We define 
\begin{equation}
f(\alpha) = \min_{\beta \in [0,1]} \max \left\{ \frac{11}{8} \alpha + \frac{1}{4}, \frac{11}{8} \alpha + \frac{3}{8} \alpha \beta, \frac{3}{2} \alpha - \frac{1}{4} \beta + \frac{1}{4}\right\}.
\end{equation}
Then, for all $i = 1, \ldots, k$ we have $$\alpha \lVert q^{\star}_i - q_i\rVert_1 + l(T_i) \le f(\alpha) \cdot l(T^{\star}_i).$$

\end{lemma}
\begin{proof}
We perform a case distinction corresponding to the Lemmas~\ref{lemma1}, \ref{lemma6}, and \ref{lemma7}.

\noindent\textbf{First case:} $B(P_i)$ is complete (as in Lemma~\ref{lemma1}). 
\begin{align*}
\alpha \lVert q^{\star}_i - q_i\rVert_1 + l(T_i) & ~\le~ \frac{1}{2} \alpha \cdot l(B(P_i)) + l(T_i) ~\le~ \frac{3}{2} \alpha \cdot l(T^{\star}_i) + (1-\frac{1}{2}\alpha) \lVert q_i - a_i \rVert_1 \\
&~\le~ \frac{3}{2} \alpha \cdot l(T^{\star}_i) + \frac{1}{4} (1-\frac{1}{2}\alpha) l(B(P_i)) ~\le~ \left(\frac{11}{8} \alpha + \frac{1}{4} \right) l(T^{\star}_i). \\
\end{align*}
where the second inequality follows from Lemma~\ref{lemma1} and the
third inequality from the fact that the Steiner tree $T'_i$ insects
the vertical line trough the origin inside the bounding box.

\noindent\textbf{Second case:} $B(P_i)$ is incomplete,  ${q^{\star}_i}_x <{q_i}_x$, and ${q^{\star}_i}_y <{q_i}_y$ (as in Lemma~\ref{lemma6}).

By the maximality of $T^{\star}_{top}$,  $q^{\star}_i$ has at least two incident edges in $T^{\star}_i$.
Since there in no point in $P_i$ which is lower-left of $q_i$, we can assume w.l.o.g., that  $q^{\star}_i$ has exactly one incident edge in the upper and one in the right direction.
The first one intersects $\{(x,y) : y = {q_i}_y\}$ in $({q^{\star}_i}_x,{q_i}_y)$ and the second one intersects $\{(x,y) : x = {q_i}_x\}$ in $({q_i}_x,{q^{\star}_i}_y)$.

Now we reroute in $T^{\star}_i$ the path starting in $({q^{\star}_i}_x,{q_i}_y)$ via $q^{\star}_i$ to $({q_i}_x,{q^{\star}_i}_y)$ by a path via $q_i$ and get a Steiner tree $T'^{\star}_i$ on $P_i$ with $l(T'^{\star}_i) = l(T^{\star}_i)$.
Since $T'^{\star}_i$ is a Steiner tree on $P_i \cup \{q_i\}$ we get $l(T_i) \le \alpha \cdot l(T'^{\star}_i) = \alpha \cdot l(T^{\star}_i)$.

With Lemma~\ref{lemma6} (first inequality), and  $4t_i^{\max} \le l(B(P_i)) \le l(T^{\star}_i)$ we get
\begin{align*}
\alpha \lVert q^{\star}_i - q_i\rVert_1 + l(T_i)
& ~\le~ \alpha \left(\frac{1}{2} l(T^{\star}_i) + (\frac{3}{2} \beta - \frac{1}{2})t_i^{\max}\right) + \alpha \cdot l(T^{\star}_i)\\
&~\le~ \frac{3}{2} \alpha \cdot l(T^{\star}_i) + \alpha \left(\frac{3}{2} \beta - \frac{1}{2}\right) t_i^{\max} 
~\le~ \frac{3}{2} \alpha \cdot l(T^{\star}_i) + \alpha \left(\frac{3}{8} \beta - \frac{1}{8}\right) l(B(P_i))  \\
&~\le~ \left(\frac{11}{8} \alpha + \frac{3}{8} \alpha \beta\right) l(T^{\star}_i).
\end{align*}

\noindent\textbf{Third case:}  $B(P_i)$ is incomplete and  $q^{\star}_i \in \{(x,y) \in B(P_i) : x \ge {q_i}_x \text{ or }  y \ge {q_i}_y\}$ (as in Lemma~\ref{lemma7}).

If $t_i = \beta t_i^{\max}$, then 
\begin{align*}
\alpha \lVert q^{\star}_i - q_i\rVert_1 + l(T_i) 
& ~\le~ \frac{1}{2} \alpha \cdot l(B(P_i)) + l(T_i')
 ~\le~ \frac{1}{2} \alpha \cdot l(B(P_i)) + \alpha \cdot l(T^{\star}_i) + \left(1 - \beta\right) t_i^{\max} \\
& ~\le~ \left(\frac{3}{2} \alpha - \frac{1}{4} \beta + \frac{1}{4}\right) l(T^{\star}_i),
\end{align*}
where the first inequality follows by Lemma~\ref{lemma7}. 
In the second inequality we use that the extra cost of connecting $q_i$ in  $T'_i$ is bounded by the distance  $(1 - \beta) t_i^{\max} $ between $q_i$ and  the upper  boundary of $B(P_i)$.

If $t_i < \beta t_i^{\max}$, then
\begin{align*}
\alpha \lVert q^{\star}_i - q_i\rVert_1 + l(T_i)
& ~\le~ \frac{1}{2} \alpha \cdot l(B(P_i)) + l(T_i) 
~\le~ \frac{3}{2} \alpha \cdot l(T^{\star}_i) + (1-\frac{1}{2}\alpha) \lVert q_i - a_i \rVert_1 \\
&~\le~ \frac{3}{2} \alpha \cdot l(T^{\star}_i) + \frac{1}{4} (1-\frac{1}{2}\alpha) l(B(P_i)) 
~\le~ \left(\frac{11}{8} \alpha + \frac{1}{4} \right) l(T^{\star}_i). \\
\end{align*}
The first inequality follows by Lemma~\ref{lemma7}, the second by Lemma~\ref{lemma2}, and the third since $\mnorm{q_i-a_i} \le t_i^{\max} \le \frac{1}{4}l(B(P_i))$. 
\end{proof}

\begin{theorem}
\label{thm:adjusted-box-center}
The two-level rectilinear  Steiner tree problem can by approximated by a factor of 
$$f(\alpha) = \min_{\beta \in [0,1]} \max \left\{ \frac{11}{8} \alpha + \frac{1}{4}, \frac{11}{8} \alpha + \frac{3}{8} \alpha \beta, \frac{3}{2} \alpha - \frac{1}{4} \beta + \frac{1}{4}\right\}$$
using an $\alpha$-factor approximation algorithm for rectilinear Steiner trees as a subroutine. 
\end{theorem}

\begin{proof}
Let $T$ be the two-level Steiner tree computeted using adjusted center points as connection points and let $T^{\star} = \{T^{\star}_{top}, T^{\star}_1, \ldots, T^{\star}_k\}$ be a minimum two-level Steiner tree with $T^{\star}_{top}$ 
as large as possible so that  all $q^{\star}_i\in B(P_i)$ for all $i=1,\dots,k$. Then
\begin{align*} 
l(T) & = l(T_{top}) + \sum_{i = 1}^{k} l(T_i) 
~\le~ \alpha \cdot l(T^{\star}_{top}) + \sum_{i = 1}^{k} \alpha \lVert q^{\star}_i - q_i\rVert_1 + \sum_{i = 1}^{k} l(T_i) \\ 
& ~\le~ \alpha \cdot l(T^{\star}_{top}) + f(\alpha) \sum_{i = 1}^{k} l(T^{\star}_i)
 ~\le~ f(\alpha) \cdot l(T^{\star}).
\end{align*}
The first inequality follows since $E\left(T^{\star}_{top}\right) \cup (\bigcup_{i = 1}^{k} \{q^{\star}_i,q_i\})$ 
covers a Steiner tree for  $\{q_1, \ldots, q_k\}$ and the third inequality follows by Lemma~\ref{lemma8}.
\end{proof}

We get an approximation ratio of $\frac{11}{8} \alpha + \frac{1}{4}$ when all bounding boxes are complete.

\begin{cor}
There is a $2.37$-factor approximation algorithm with runtime $\mathcal{O}(n \log n)$ for the two-level rectilinear  Steiner tree problem.
\end{cor}
\begin{proof}
  A minimum rectilinear Steiner tree on $l$ terminals in the plane can be
  approximated by a factor $\alpha= 1.5$ computing  a minimum spanning tree in the Delaunay
  triangulation  in  $\mathcal{O}(l\log l)$ time \cite{smith}.
  Now we apply  Theorem~\ref{thm:adjusted-box-center}, where $f(\alpha) = \frac{123}{52} < 2.37$  is determined by $\beta = \frac{7}{13}$.
\end{proof}

\begin{cor}
There is an $1.63$-factor approximation algorithm for the two-level rectilinear Steiner tree problem.
\end{cor}
\begin{proof}
  We approximate the rectilinear Steiner trees  in Theorem~\ref{thm:adjusted-box-center} using Arora's approximation scheme \cite{arora}.
  Choosing $0< \epsilon < 0.003$, we get the claimed approximation factor, where $f(\alpha)$ is determined by $\beta := \frac{3}{5}$.
\end{proof}

\subsection{Small Top-Level Trees}

The adjusted box center algorithm computes the connection points based
on each partition individually ignoring the structure of the top-level
tree.  
If there is small box containing one terminal from each set $P_1,\dots,P_k$,
we can give a better approximation ratio independent of the structure of the partition sets.
We define the {\it top-level   bounding box} $B_{top}(P_1,\dots,P_k)$ as the smallest axis-parallel rectangle
containing at least one point $q_i \in P_i$ for each $i = 1, \ldots k$.
It is a simple exercise to see that $B_{top}(P_1,\dots,P_k)$  can be computed in $\mathcal{O}(n^3)$.

We choose connection points $q_i \in P_i \cap B_{top}(P_1,\dots,P_k)$ and compute the two-level Steiner tree
$(T_{top}, T_1, \ldots T_k)$ by an $\alpha$-factor approximation algorithm for rectilinear Steiner trees.
Let $(T_{top}^{\star}, T_1^{\star}, \ldots T_k^{\star})$ be an optimum solution.
In \cite{ulrich} it was shown that
$U(k) := \frac{\lceil\sqrt{k-2}\rceil}{2} + \frac{3}{4}$ is an  upper bound for the ratio of the length of a minimum 
rectilinear Steiner tree on $k$ terminals and their bounding box. Then
\begin{align*} 
l(T) & = l(T_{top}) + \sum_{i = 1}^{k} l(T_i) ~\le~ U(k) \cdot l(B_{top}(P_1,\dots,P_k)) + \alpha \sum_{i = 1}^{k} l(T^{\star}_i).\\ 
\end{align*}
Now if $l(B_{top}(P_1,\dots,P_k))$ is small, the approximation factor for the two-level Steiner tree 
is essentially  dominated by $\alpha$.

\small
\bibliographystyle{abbrv}

\begin{thebibliography}{99}

\bibitem{alvarez-etal:2014}
E. Alvarez-Miranda, I. Ljubic, S. Raghavan, and P. Toth:
\newblock Recoverable Robust Two-Level Network Design Problem.
\newblock {\em INFORMS Journal on Computing}, to appear, 2014.

\bibitem{arora}
S. Arora:
\newblock Polynomial time approximation schemes for Euclidian Traveling Salesman and other geometric problems.
\newblock {\em Journal of the ACM}, 45(5): 753--782, 1998.

\bibitem{baiou}
 M. Baïou and F. Barahona.
\newblock A polyhedral study of a two level facility location model . 
\newblock {\em RAIRO - Operations Research}, 48:  153--165, 2014.

\bibitem{byrka}
J. Byrka and B. Rybicki.
\newblock Improved LP-Rounding Approximation Algorithm for k-level Uncapacitated Facility Location.
\newblock {\em Proceedings of the 39th  International Colloquium onAutomata, Languages, and Programming} (ICALP '12), 157--169, 2012.



\bibitem{balakrishnan:1991b}
A. Balakrishnan, T. L. Magnanti, and P. Mirchandani:
\newblock A dual-based algorithm for multi-level network design. 
\newblock {\em Management Science}, 40(5): 567--581,  1994. 

\bibitem{ulrich}
U. Brenner, J. Vygen:
\newblock Worst-Case Ratios of Networks in the Rectilinear Plane. 
\newblock {\em Networks}. 38, 126--139, 2001.

\bibitem{current}
J. R. Current, C. S. ReVelle, and J. L. Cohon:
\newblock The hierarchical network design problem.
\newblock {\em European Journal of Operational Research},  27(1), 57--66, 1986.

\bibitem{edmonds+karp:1972}
J. Edmonds and R.M. Karp:
\newblock Theoretical Improvements in Algorithmic Efficiency for Network Flow Problems.
\newblock {\em Journal of the ACM}, 19(2), 248--264 ,1972.
 

\bibitem{hanan}
M. Hanan:
\newblock On Steiner's problem with rectilinear distance.
\newblock {\em SIAM Journal on Applied Mathematics}, 14(2): 255--265, 1966.


\bibitem{menger}
K. Menger:
\newblock Zur allgemeinen Kurventheorie.
\newblock {\em Fundamenta Mathematicae}, 10: 96--115, 1927.

\bibitem{safra}
S. Safra and O. Schwartz:
\newblock On the complexity of approximating TSP with neighborhoods and related problems.
\newblock {\em Computational Complexity}, 14(4): 281--307, 2006.

\bibitem{smith}
J.M. Smith, D.T. Lee, and J.S. Liebman:
\newblock An $\mathcal{O}(n \log n)$ heuristic for Steiner minimal tree problems on the euclidean metric.
\newblock {\em Networks},   11(1): 23--39, 1981.

\bibitem{snyder}
T.L. Snyder:
\newblock On the exact location of  Steiner points in general dimension.
\newblock {\em SIAM Journal on Computing},  21(1): 163--180, 1991.

\bibitem{xiang-etal:2010}
H. Xiang, H. Ren, L. Trevillyan, L. Reddy, R. Puri, and M. Cho:
\newblock Logical and physical restructuring of fan-in trees. 
\newblock {\em Proceedings of the 19th International Symposium on Physical design} (ISPD '10), 67--74, 2010.

\bibitem{xiang-etal:2013}
 H. Xiang, L. Reddy, L. Trevillyan, and  R. Puri:
 \newblock Depth Controlled Symmetric Function Fanin Tree Restructure.
 \newblock {\em Proceedings of the International Conference on Computer-Aided Design} (ICCAD '13), 585--591, 2013.




\end{thebibliography}

\end{document}